\documentclass[12pt]{amsart}

\usepackage{a4wide}
\usepackage{amssymb}
\usepackage{graphicx}
\usepackage{epstopdf}

\newcommand{\cF}{\mathcal{F}}
\newcommand{\bE}{\mathbb{E}}
\newcommand{\bP}{\mathbb{P}}

\newtheorem{thm}{Theorem}[section]

\theoremstyle{definition}

\newtheorem{rem}[thm]{Remark}

\numberwithin{equation}{section}

\begin{document}

\title[Merton portfolio problem with one indivisible asset]{Merton portfolio problem with one indivisible asset}

\author{Jakub Trybu{\l}a}

\address{\noindent Jakub Trybu{\l}a, \newline \indent Institute of Mathematics \newline \indent Faculty of Mathematics and Computer Science \newline \indent  Jagiellonian University in Krakow \newline \indent{\L}ojasiewicza  6 \newline \indent 30-348 Krak{\'o}w, Poland}

\email{jakub.trybula@im.uj.edu.pl}



\begin{abstract}
In this paper we consider a modification of the classical Merton portfolio optimization problem. Namely, an investor can trade in financial asset and consume his capital. He is additionally endowed with a one unit of an indivisible asset which he can sell at any time. We give a numerical example of calculating the optimal time to sale the indivisible asset, the optimal consumption rate and the value function.
\end{abstract}

\maketitle

\section{Introduction}

The earliest application of stochastic control theory to finance was Merton's portfolio optimization problem which is formulated as follows. An investor has a positive capital and access to the market where he can freely buy and sale two assets: a share and a bond. He can also consume arbitrary part of his wealth. At any time the investor decides which proportion of his wealth to invest in risky asset and which part of his capital to consume. The purpose is to find proportion of wealth and consumption rate which maximizes investor's performance criterion. This problem has a well known solution. For more details see e.g. \cite{2,3,5} and also Section \ref{sec:1}.  

In this work we modify Merton's portfolio optimization problem. Namely, we assume that our investor has no access to a bank account but is endowed with a one unit of an indivisible financial asset which he can sell at any time. By indivisible asset we mean the asset which investor cannot divide and sale only just part of it. It can be for example a house, an apartment, a work of art, a car or simply one unit of share. The money received from the sale of this asset is immediately invested. 

The aim is to find the time to sale the indivisible asset and the consumption rate which maximizes the investor's performance functional. 

It turns out that the problem can be reduced to a one dimensional differential equation with a boundary conditions. However, if we add additional financial assets the problem becomes multidimensional and much harder to solve.

Investor problems are investigated in many particular cases. Typically, see e.g. \O{}ksendal and Sulem \cite{4} and original paper by Merton \cite{3}, one is able to find explicite solution to the HJB equation. In our case we are only able to find numerically solution assuming its smoothness. The problem if given HJB equation has a smooth solution or if the value function is smooth, has been only solved under the uniform ellipticity condition see e.g. \cite{2}, p. 165, p. 171, and references therein and \cite{5}, p. 51, for parabolic case. Unfortunately, our problem fails to be uniformly elliptic.

The problem of optimal selling was also formulated and studied in \cite{1}. The authors considered an investor with power utility who owns a single unit of an indivisible asset and who wishes to choose the optimum time to sell this asset. They showed that the optimal strategy is to sell the indivisible asset the first time that its value exceeds a certain proportion of the investor's wealth.

\section{Classical Merton problem}\label{sec:1}

Let $(\Omega,\cF,(\cF_{t}),\bP)$ be a filtered probability space. Define $\left(X_{t}\right)$ as the wealth process of the investor. Without any loss of generality we assume that the consumption rate is of the form $c_{t}X_{t}$, where the $(\cF_{t})$-adapted process $(c_{t})$ will be called the consumption. Then $X_{t}$ changes according to the stochastic differential equation
\[
\left\{ \begin{aligned}
dX_{t}&=\left(\mu -c_{t}\right)X_{t}dt+\sigma X_{t}dW_{t},\quad t>0,\\
X_{0}&=x>0,
\end{aligned}\right.
\]
where $\mu$, $\sigma>0$ are constants and $\left(W_{t}\right)$ is a Brownian motion. Let $\alpha\in(0,1)$ and the rate $\beta$ satisfies
\begin{equation}\label{beta1}
\beta>\max\left\{0,\alpha\mu+\frac{1}{2}\alpha(\alpha-1)\sigma^{2}\right\}.
\end{equation}
The objective is to find a consumption $\hat{c}$ such that
\[
J^{(\hat{c})}(x):=\sup J^{(c)}(x),
\]
where
\[
J^{(c)}(x):=\bE^{(x)}\left[\int_{0}^{+\infty}e^{-\beta t}\frac{(c_{t}X_{t})^{\alpha}}{\alpha}dt\right],
\]
and the supremum is taken over all admissible strategies $c$. We call 
\[
V(x):=J^{(\hat{c})}(x),
\]
the value function and interpret $J^{(c)}(x)$ as expected value of investor's consumption rate discounted at rate $\beta$. This classical Merton's problem has a well known solution, see e.g. \cite{2}, p. 168. Namely,
\[
V(x)=Ax^{\alpha}\quad\text{and}\quad \hat{c}_{t}=(A\alpha)^{\frac{1}{\alpha-1}},
\]
where
\begin{equation}\label{A}
A=\frac{1}{\alpha}\left[\frac{\beta-\mu\alpha-\frac{1}{2}\alpha(\alpha-1)\sigma^{2}}{1-\alpha}\right]^{\alpha-1}.
\end{equation}
Assumption \eqref{beta1} guarantees that the constant $A$ is well defined.

\section{Modified Merton problem}

Assume now that our investor is additionally endowed with a one unit of an indivisible asset. The value of this asset is given by the equation 
\[
\left\{ \begin{aligned}
dY_{t}&=\tilde{\mu} Y_{t}dt+\tilde{\sigma} Y_{t}d\tilde{W}_{t},\quad t>0,\\
Y_{0}&=y>0,
\end{aligned}\right.
\]
where $\tilde{\mu}$, $\tilde{\sigma}>0$ are constants and $(\tilde{W}_{t})$ is an independent of $(W_{t})$ one-dimensional Brownian motion. The investor can sale the indivisible asset at arbitrary time $\tau\geq 0$ and must immediately invest the money received from this sale in the risky asset. We assume that after the indivisible asset is sold the agent consumes his capital in an optimal fashion. It means that investor behaves in accordance with solution to the Merton portfolio optimization problem on interval $[\tau,+\infty)$ with initial value $x=X_{\tau}=X^{-}_{\tau}+Y_{\tau}$, where $X^{-}_{\tau}$ is value of investor's wealth process immediately prior to the sale of the indivisible asset. 

The aim is to find a stopping time $\hat{\tau}$ and a consumption $\hat{c}$ such that
\[
J^{(\hat{\tau},\hat{c})}(x,y):=\sup J^{(\tau,c)}(x,y),
\]
where
\[
J^{(\tau,c)}(x,y):=\bE^{(x,y)}\left[\int_{0}^{\tau}e^{-\beta t}\frac{(c_{t}X_{t})^{\alpha}}{\alpha}dt+e^{-\beta\tau}A\left(X_{\tau}+Y_{\tau}\right)^{\alpha}\chi_{\{\tau<+\infty\}}\right],
\]
$A$ is given by \eqref{A} and the supremum is taken over all admissible $\tau$ and $c$. We interpret $J^{(\tau,c)}(x,y)$ as the expected value of investor's consumption rate discounted at rate $\beta$ and increased by premium for sale of the indivisble asset at time $\tau$. 

\begin{rem}
Note that if $\beta<\alpha\tilde{\mu}+\frac{1}{2}\alpha(\alpha-1)\tilde{\sigma}^{2}$, then the problem has a trivial solution. Namely, assume that the time to sale indivisible asset is of the form $\tau_{n}=n$. Then 
\[
J^{(n,c)}(x,y)=\bE^{(x,y)}\left[\int_{0}^{n}e^{-\beta t}\frac{(c_{t}X_{t})^{\alpha}}{\alpha}dt+e^{-\beta n}A(X_{n}+Y_{n})^{\alpha}\right]
\]
\[
\geq \bE^{(x,y)}\left[e^{-\beta n}AY^{\alpha}_{n} \right]=Ay^{\alpha}\exp\left\{\left(-\beta +\alpha\tilde{\mu}+\frac{1}{2}\alpha(\alpha-1)\tilde{\sigma}^{2}\right)n\right\}.
\]
Hence, we have
\[
\lim_{n\to+\infty}J^{(n,c)}(x,y)=+\infty,
\]
so the value function is infinite and there is no optimal stopping time. It is always better to postpone the decision to sell. Therefore (see also \eqref{beta1}), we assume that
\begin{equation}\label{beta2}
\beta>\max\left\{0,\alpha\mu+\frac{1}{2}\alpha(\alpha-1)\sigma^{2},\alpha\tilde{\mu}+\frac{1}{2}\alpha(\alpha-1)\tilde{\sigma}^{2}\right\}.
\end{equation}
\end{rem}

\section{Solution by the Hamilton-Jacobi-Bellman verification theorem }\label{sec:4}

In this section using the verification theorem we reduce the problem to ordinary differential equation with boundary conditions.

\begin{thm}\label{twr:1} Let us assume that $A$ is a constant given by \eqref{A} and $\beta$ satisfies \eqref{beta2}. Suppose that 
\begin{equation}\label{coe:1}
(\alpha-1)\sigma^{2}<\tilde{\mu}-\mu <\frac{1}{2}(\alpha-1)\left(\sigma^{2}-\tilde{\sigma}^{2}\right).
\end{equation}
Moreover, suppose that $K\in C^{2}(0,+\infty)$ and $z^{*}>0$ are such that for every $z\in \left(z^{*},+\infty\right)$,
\begin{align}
A\left(z+1\right)^{\alpha}<&K(z)<A\left[z^{\alpha}+\left(z+1\right)^{\alpha}\right],\label{e:3}\\
0<&K'(z)<\frac{\alpha}{z}K(z) \label{e:4}
\end{align}
and
\[
\frac{1}{2}\left(\sigma^{2}+\tilde{\sigma}^{2}\right)z^{2}K''\left(z\right)+\left[\mu-\tilde{\mu}+(1-\alpha)\tilde{\sigma}^{2}\right]zK'\left(z\right)
\]
\begin{equation}\label{e:1}
+\left[-\beta+\alpha\tilde{\mu}+\frac{1}{2}\alpha(\alpha-1)\tilde{\sigma}^{2}\right]K\left(z\right)+\frac{1-\alpha}{\alpha}\left(K'\left(z\right)\right)^{\frac{\alpha}{\alpha-1}}=0,
\end{equation}
with boundary conditions
\begin{equation}\label{e:2}
K(z^{*})=A\left(z^{*}+1\right)^{\alpha}\qquad\text{and}\qquad K'(z^{*})=A\alpha\left(z^{*}+1\right)^{\alpha-1}.
\end{equation}
Then
\begin{enumerate}
\item Free boundary $z^{*}$ satisfies 
\[
z^{*}\leq \frac{\mu-\tilde{\mu}+\frac{1}{2}(\alpha-1)(\sigma^{2}-\tilde{\sigma}^{2})}{\tilde{\mu}-\mu -(\alpha-1)\sigma^{2}}.
\]
\item Optimal time to sale the indivisible asset is of the form
\[
\hat{\tau}=\inf\left\{t\geq 0\colon \frac{X_{t}}{Y_{t}}\leq z^{*}\right\}.
\]
\item Optimal consumption rate is of the form
\[
\hat{c}_{t}X_{t}=\left\{ \begin{array}{ll}
Y_{t}\left(K\left(\frac{X_{t}}{Y_{t}}\right)\right)^{\frac{1}{\alpha-1}},& if \ t\in[0,\hat{\tau}),\\
(A\alpha)^{\frac{1}{\alpha-1}}X_{t},&if \ t\in[\hat{\tau},+\infty).\\
\end{array} \right.
\]
\item The value function is given by
\[
V(x,y)=\left\{ \begin{array}{ll}
A(x+y)^{\alpha},& if \ \frac{x}{y}\in [0,z^{*}),\\
y^{\alpha}K\left(\frac{x}{y}\right),& if \ \frac{x}{y}\in[z^{*},+\infty).\\
\end{array} \right.
\]
\end{enumerate}
\end{thm}

\begin{proof} Let $\mathcal{O}:=(0,+\infty)\times(0,+\infty)$ and define
\begin{align}
L^{c}\varphi(x,y):=& (\mu-c)x\varphi_{x}(x,y)+\tilde{\mu}y\varphi_{y}(x,y) \\
&+\frac{1}{2}(\sigma x)^{2}\varphi_{xx}(x,y)+\frac{1}{2}(\tilde{\sigma}y)^{2} \varphi_{yy}(x,y),  \notag
\end{align}
for all functions $\varphi\colon\mathcal{O}\to\mathbb{R}$ which are twice differentiable at $(x,y)$. According to the verification theorem (see e.g. \cite{4}, p. 62) we are looking for the continuation region $D\subset \mathcal{O}$ and function $\psi\colon \mathcal{O}\to\mathbb{R}$ of class $C^{2}(\mathcal{O}\backslash\partial D)\cap C^{1}(\partial D)$ such that 
\begin{enumerate}
\item[\textbf{I.}] For all $(x,y)\in\mathcal{O}\backslash\bar{D}$ and for all $c\geq 0$, 
\begin{equation}\label{b:1}
\psi(x,y)=A\left(x+y\right)^{\alpha},
\end{equation}
and 
\begin{equation}\label{b:2}
-\beta\psi(x,y)+L^{c}\psi(x,y)+\frac{(cx)^{\alpha}}{\alpha}\leq 0.
\end{equation}
\item[\textbf{II.}] For $(x^{*},y^{*})\in\partial D$, 
\begin{equation}\label{b:5}
\psi(x^{*},y^{*})=A\left(x^{*}+y^{*}\right)^{\alpha},
\end{equation}
\begin{equation}\label{b:6}
\psi_{x}(x^{*},y^{*})=A\alpha\left(x^{*}+y^{*}\right)^{\alpha-1}\ \text{and }\ \psi_{y}(x^{*},y^{*})=A\alpha\left(x^{*}+y^{*}\right)^{\alpha-1}.
\end{equation}
\item[\textbf{III.}] For all $(x,y)\in D$,
\begin{equation}\label{e:6}
\psi(x,y)>A\left(x+y\right)^{\alpha}
\end{equation}
and 
\begin{equation}\label{b:7}
-\beta\psi(x,y)+\sup_{c\geq 0}\left\{L^{c}\psi(x,y)+\frac{(cx)^{\alpha}}{\alpha}\right\}=0.
\end{equation}
\end{enumerate}

The key idea of the proof relies on the observation that our problem can be reduced to one-dimensional. Namely, we expect $\psi(x,y)$ to be of the form
\begin{equation}\label{A:2}
\psi(x,y):=y^{\alpha}K\left(\frac{x}{y}\right)\quad \text{for all}\ (x,y)\in \bar{D},
\end{equation}
where $K\colon (0,+\infty)\to\mathbb{R}$ is unknown function. Moreover, let
\begin{equation}\label{D}
D:=\left\{(x,y)\in \mathcal{O}\colon \frac{x}{y}\in\left(\frac{x^{*}}{y^{*}},+\infty\right)\right\},
\end{equation}
where $x^{*}/y^{*}=:z^{*}$ is unknown free boundary. 

Summarizing, we need to find the free boundary $z^{*}>0$ and the function $K$ which satisfy the three conditions above. 

Now we verify the conditions \textbf{I}, \textbf{II} and \textbf{III} for the function \eqref{A:2} and continuation region \eqref{D}. 

\noindent\textbf{Ad. I.} In \textbf{I} the function $K$ does not appear but if we put \eqref{b:1} into inequality \eqref{b:2} we get a condition for the free boundary. Namely, for all
\[
(x,y)\in\mathcal{O}\backslash\bar{D}=\left\{(x,y)\in \mathcal{O}\colon 0<\frac{x}{y}<z^{*}\right\},
\]
we have that
\[
-\beta\psi(x,y)+\sup_{c\geq 0}\left\{L^{c}\psi(x,y)+\frac{(cx)^{\alpha}}{\alpha}\right\}\leq 0,
\]
which can be stated equivalently as follows
\[
-\mu-\frac{1}{2}(\alpha-1)\sigma^{2}+\mu\frac{x}{x+y} +\tilde{\mu}\frac{y}{x+y}
\]
\[
+\frac{1}{2}(\alpha-1)\sigma^{2}\frac{x^{2}}{(x+y)^{2}}+\frac{1}{2}(\alpha-1)\tilde{\sigma}^{2}\frac{y^{2}}{(x+y)^{2}}\leq 0.
\]
Putting $z:=x/y$ we obtain
\[
\left[\tilde{\mu}-\mu -(\alpha-1)\sigma^{2}\right]z+\tilde{\mu}-\mu+\frac{1}{2}(\alpha-1)(\tilde{\sigma}^{2}-\sigma^{2})\leq 0.
\]
If $z\to 0$, then assumption \eqref{coe:1} guarantees that above inequality is satisfied. Moreover, we have
\[
z^{*}\leq \frac{\mu-\tilde{\mu}+\frac{1}{2}(\alpha-1)(\sigma^{2}-\tilde{\sigma}^{2})}{\tilde{\mu}-\mu -(\alpha-1)\sigma^{2}}.
\]

\noindent\textbf{Ad. II.} For $(x^{*},y^{*})\in\partial D$, from conditions \eqref{b:5} and \eqref{b:6}, we have
\[
K(z^{*})=A\left(z^{*}+1\right)^{\alpha}\qquad\text{and}\qquad K'(z^{*})=A\alpha\left(z^{*}+1\right)^{\alpha-1}.
\]

\noindent\textbf{Ad. III.} For all $(x,y)\in D$, the supremum in \ref{b:7} is attained at
\begin{equation}\label{b:8}
\hat{c}=\frac{y}{x}\left(K\left(\frac{x}{y}\right)\right)^{\frac{1}{\alpha-1}}.
\end{equation}
Thus from the verification theorem we know that 
\[
\hat{c}_{t}X_{t}=Y_{t}\left(K\left(\frac{X_{t}}{Y_{t}}\right)\right)^{\frac{1}{\alpha-1}}.
\]
If we put \eqref{b:8} into HJB equation we obtain
\[
\frac{1}{2}\left(\sigma^{2}+\tilde{\sigma}^{2}\right)\frac{x^{2}}{ y^{2}}K''\left(\frac{x}{y}\right)+\left[\mu-\tilde{\mu}+(1-\alpha)\tilde{\sigma}^{2}\right]\frac{x}{y}K'\left(\frac{x}{y}\right)
\]
\[
+\left[-\beta+\alpha\tilde{\mu}+\frac{1}{2}\alpha(\alpha-1)\tilde{\sigma}^{2}\right]K\left(\frac{x}{y}\right)+\frac{1-\alpha}{\alpha}\left(K'\left(\frac{x}{y}\right)\right)^{\frac{\alpha}{\alpha-1}}=0,
\]
for all $(x,y)\in D$. Putting $z:=x/y$ we get the equation \eqref{e:1}.

Summing up we reduced our problem to ordinary differential equation with boundary conditions. 
\begin{rem}
Condition \eqref{e:6} is of the form 
\[
A\left(z+1\right)^{\alpha}<K\left(z\right)\quad \text{for all}\ z\in \left(z^{*},+\infty\right).
\]
\end{rem}
\begin{rem}
For all $(x,y)\in D$, using the solution to Merton's classical portfolio problem, we get 
\[
J^{(\tau,c)}(x,y)\leq \mathbb{E}^{(x,y)}\left[\int_{0}^{+\infty}e^{-\beta t}\frac{(c_{t}X_{t})^{\alpha}}{\alpha}dt\right]+\mathbb{E}^{(x,y)}\left[e^{-\beta\tau}A\left(X_{\tau}+Y_{\tau}\right)^{\alpha}\chi_{\{\tau<+\infty\}}\right]
\]
\[
< Ax^{\alpha}+A(x+y)^{\alpha}.
\]
For this reason the function $K$ must satisfy
\[
K\left(z\right)< A\left[z^{\alpha}+\left(z+1\right)^{\alpha}\right]\quad \text{for all}\ z\in \left(z^{*},+\infty\right).
\]
\end{rem}
\begin{rem}
The value function is increasing function of the initial conditions. It means that $\psi_{x}(x,y)>0$ and $\psi_{y}(x,y)>0$ for every $(x,y)\in D$. Therefore we obtain
\begin{equation}\label{g:3}
0<K'(z)<\frac{\alpha}{z}K(z)\quad \text{for all}\ z\in \left(z^{*},+\infty\right).
\end{equation}
Note that by \eqref{g:3} the HJB equation is well defined.
\end{rem}
\begin{rem}
From the verification theorem the optimal stopping time is of the form
\[
\hat{\tau}=\inf\left\{t\geq 0\colon \left(X_{t},Y_{t}\right)\notin D\right\}.
\]
It means that the optimal time to sell the indivisible asset is
\[
\hat{\tau}=\inf\left\{t\geq 0\colon \frac{X_{t}}{Y_{t}}\leq z^{*}\right\}.
\]
\end{rem}

\end{proof}

\section{Numerical example}\label{sec:5}

In this section we give a numerical example of calculating the optimal time to sale the indivisible asset, the optimal consumption rate and the value function. 

Let us assume that $\beta$, $\mu$, $\sigma$, $\tilde{\mu}$, $\tilde{\sigma}$ and $\alpha$ satisfy the assumptions of Theorem \ref{twr:1}. Then for every $z^{*}\in (0,\tilde{z}]$ we can numerically solve the HJB equation with the boundary conditions. The numerical observations (using Mathematica 9.0) suggest that  
\begin{enumerate}
\item There is only one $z^{*}=\hat{z}$ for which the equation (\ref{e:1}) with the boundary conditions (\ref{e:2}) has a global solution $K(z)$ which satisfies the conditions (\ref{e:3}) and (\ref{e:4}). 
\item For $z^{*}\in\left(0,\hat{z}\right)$ the conditions (\ref{e:3}) and (\ref{e:4}) are not satisfied, that is either the solution grows too fast or its derivative does not satisfy the constraints in (\ref{e:4}).
\item For $z^{*}\in\left(\hat{z},\tilde{z}\right]$ the solution is only local, that is the acceleration goes to $-\infty$ or the conditions (\ref{e:3}) and (\ref{e:4}) are not satisfied. 
\end{enumerate}
Now we present numerical solution for parameters
\[
\beta=2,\quad \mu=1,\quad \sigma=1,\quad \tilde{\mu}=0.5,\quad \tilde{\sigma}=0.5\quad \text{and}\quad\alpha=\frac{1}{3}.
\]
Note that the assumptions of Theorem \ref{twr:1} are satisfied. Then we have that $A=1.56006$, $\tilde{z}=1.5$ and from numerical observations we get
\[
\hat{z}\approx 1.3169624.
\]
Let $K_{i}(z)$ be the solution to equation (\ref{e:1}) with boundary conditions
\[
K_{i}(z^{*}_{i})=A\left(z^{*}_{i}+1\right)^{\alpha}\quad\text{and}\quad K'_{i}(z^{*}_{i})=A\alpha\left(z^{*}_{i}+1\right)^{\alpha-1}.
\] 
From assumption $(\ref{e:3})$ we have that $K_{i}(z)$ must satisfy the following inequality
\[
B^{l}(z)<K_{i}(z)<B^{u}(z),
\]
where
\[
B^{l}(z)=A\left(z+1\right)^{\alpha}\quad\text{and}\quad B^{u}(z)=A\left[z^{\alpha}+\left(z+1\right)^{\alpha}\right].
\]
Moreover, from $(\ref{e:4})$ the function $K'_{i}(z)$ must live between two curves $C^{l}_{i}(z)$ and $C^{u}_{i}(z)$, where
\[
C^{l}_{i}(z)=0\quad\text{and}\quad C^{u}_{i}(z)=\frac{\alpha}{z}K_{i}(z).
\]
In the figure below we present the functions $K_{i}(z)$ for $i=1,\ldots,5$, where
\[
z^{*}_{1}=1,\quad z^{*}_{2}=1.31696,\quad z^{*}_{3}=\hat{z},\quad z^{*}_{4}=1.31697,\quad\text{and}\quad z^{*}_{5}=\tilde{z}.
\]
\begin{center}
\includegraphics[width=13cm,height=8cm]{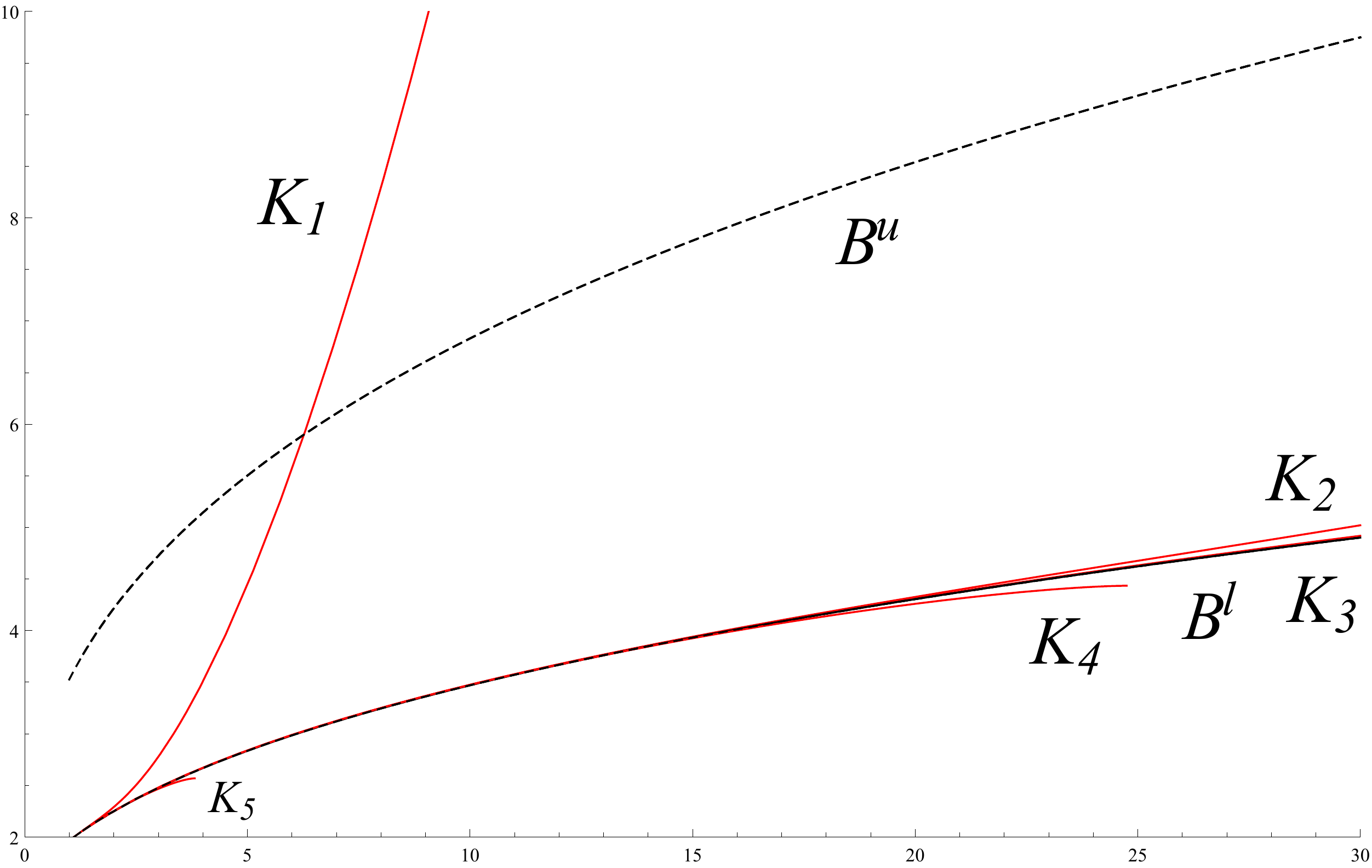}
\end{center}
Note that the function $K_{3}(z)$ almost coincides with the function $B^{l}(z)$. The difference between them is visible in the next figure on a different scale. 
\begin{center}
\includegraphics[width=13cm,height=8cm]{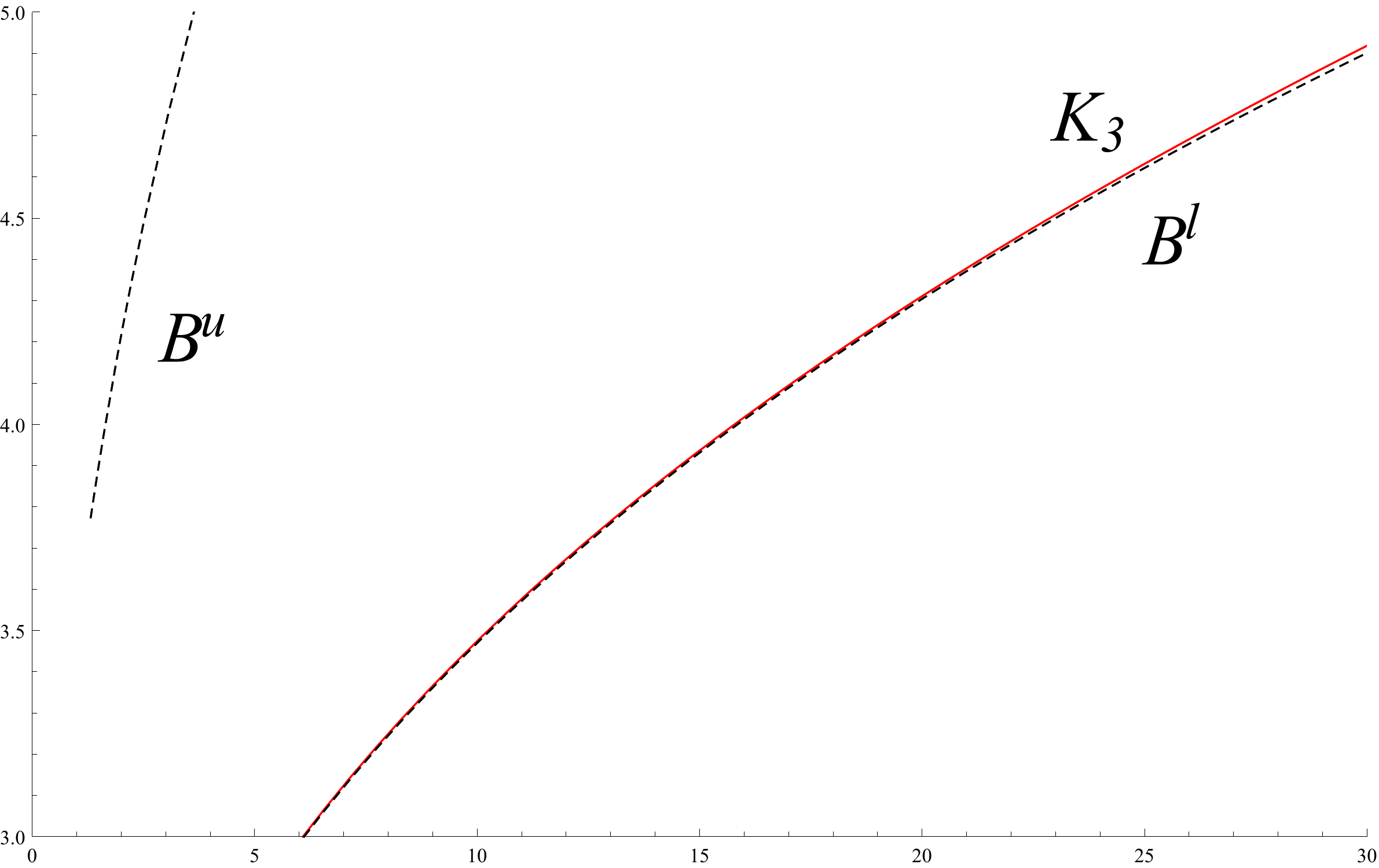}
\end{center}
Further, we present $K'_{3}(z)$.
\begin{center}
\includegraphics[width=13cm,height=8cm]{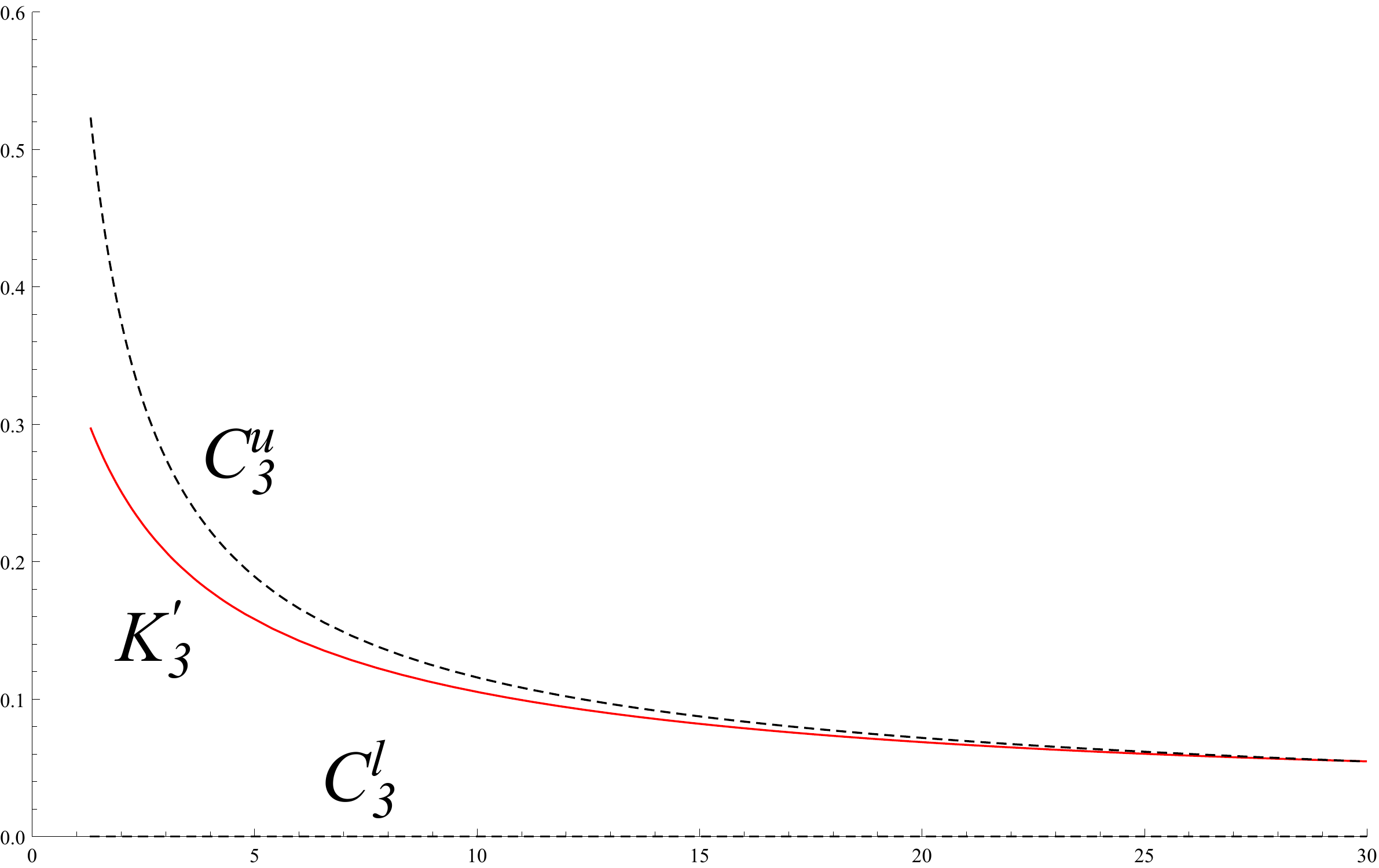}
\end{center}
Finally, in the next figure we show $K'_{2}(z)$.
\begin{center}
\includegraphics[width=13cm,height=8cm]{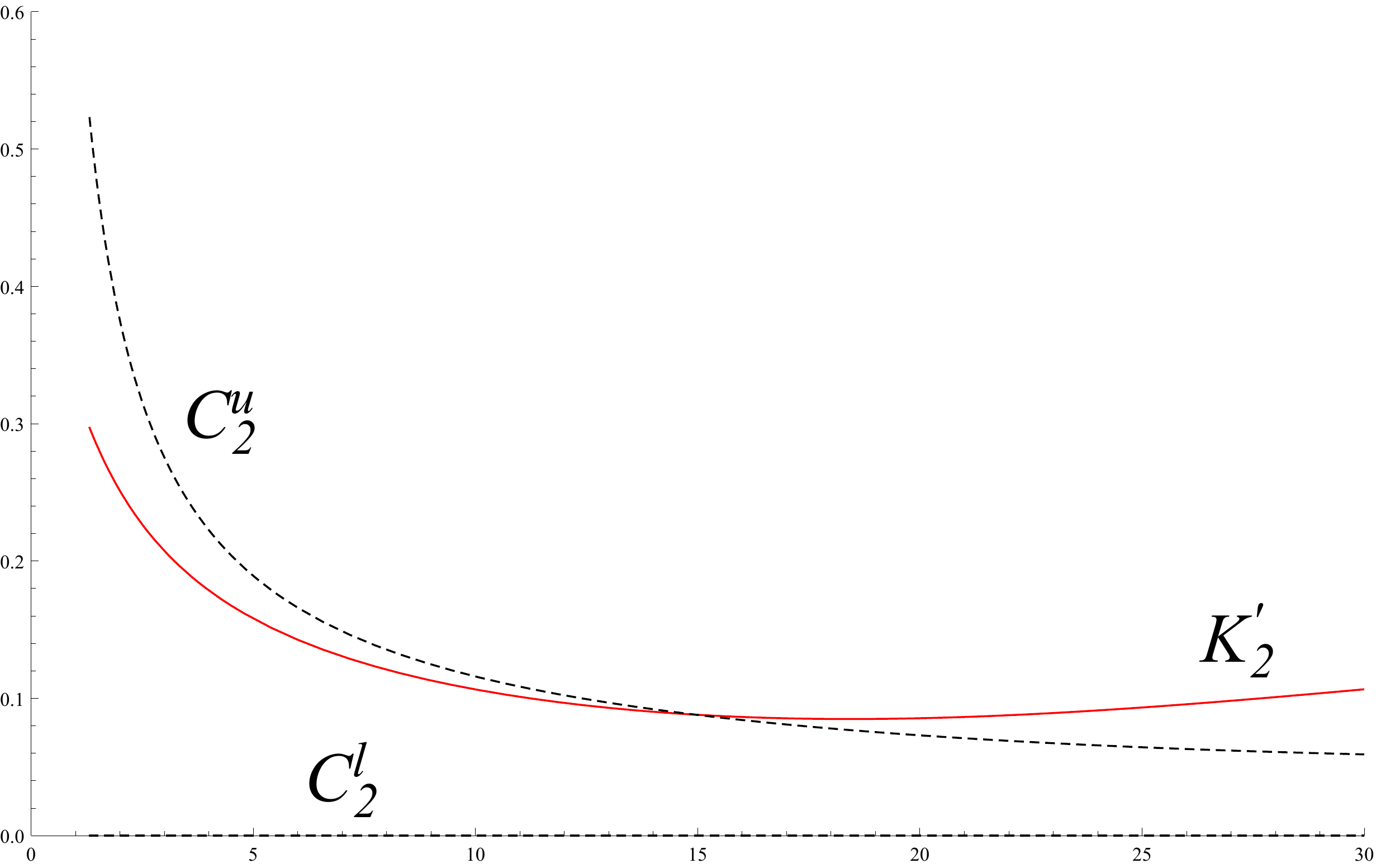}
\end{center}

\subsection*{Acknowledgements}
The author is very grateful to Professor Szymon Peszat for his support, meetings, helpful comments and suggestions.

\end{document}